\documentclass[
  notitlepage,
  twocolumn,
  preprintnumbers,
  longbibliography,
  superscriptaddress,
  aps,
  prb,
  amsmath,
  amssymb,
]{revtex4-2}

\usepackage{CJKutf8}
\newcommand{\CN}[1]{\begin{CJK*}{UTF8}{gbsn}#1\end{CJK*}} 
\usepackage[justification=raggedright,singlelinecheck=false]{caption}
\usepackage{subcaption}
\usepackage[margin=1in]{geometry} 

\usepackage{booktabs} 
\usepackage{physics}
\usepackage{tikz}
\usepackage{tkz-euclide}
\usetikzlibrary{positioning}
\usetikzlibrary{decorations.markings}
\usepackage{ifthen}
\usepackage{braket}
\usepackage[normalem]{ulem}

\usepackage{amsmath}
\usepackage{amsthm, amssymb}
\usepackage{slashed} 
\usepackage{graphicx}
\usepackage{xcolor}
\usepackage{bm}
\usepackage[
  pdfpagelabels,
  plainpages=false,
  bookmarks=true,
  colorlinks=true,
  linkcolor=red,
  urlcolor=blue,
  citecolor=blue
]{hyperref}
\usepackage{dsfont}
\usepackage{changepage}
\usepackage{array}
\usepackage{bbm}

\renewcommand{\vec}[1]{\bm{\mathbf{#1}}}
\newcommand{\Z}{\mathbb{Z}}

\newcommand{\vac}{\text{vac}}
\newcommand{\up}{{\scriptscriptstyle\uparrow}}
\newcommand{\down}{{\scriptscriptstyle\downarrow}}
\newcommand{\shortarrow}{\hspace{-0.5mm}\mathrel{\scalebox{0.7}{$\rightarrow$}}\hspace{-0.5mm}}
\newcommand{\loc}{\text{loc}}
\newcommand{\LE}{\texttt{LE}}
\newcommand{\RE}{\texttt{RE}}

\newtheorem{theorem}{Theorem}

\theoremstyle{definition}

\theoremstyle{remark}

\usepackage[ruled,vlined]{algorithm2e}
\begin{document}

\title{Algorithms for variational Monte Carlo calculations of fermion projected entangled pair states in the swap gates formulation and the detailed balance of tensor network sequential sampling}

\author{Yantao Wu (\CN{武琰涛})}
\email{yantaow@iphy.ac.cn}
\affiliation{Institute of Physics, Chinese Academy of Sciences, Beijing 100190, China}

\author{Zhehao Dai (\CN{戴哲昊})}
\affiliation{Department of Physics and Astronomy, University of Pittsburgh, PA 15213, USA}

\date{\today}

\begin{abstract}
In recent years, the variational Monte Carlo (VMC) calculations of projected entangled pair states (PEPS) has emerged as a competitive method for computing the ground states of many-body quantum systems. 
This method is particularly important for fermion systems where sign problems are abundant. 
We derive and explain the algorithms for the VMC calculations of fermion PEPS in the swap gates formulation.  
As a separate key result, we prove the detailed balance of sequential sampling of tensor networks. 
\end{abstract}

\maketitle
\section{Introduction}
Numerical computation of ground states of two-dimensional many-body fermion systems is a fundamental problem in condensed matter physics. 
Due to the fermion exchange, sign problems are abundant in quantum Monte Carlo simulations \cite{loh1990sign} for fermions, making the method difficult for a broad range of signful systems.  
Based on the variational principle of energy, variational calculations using a wavefunction ansatz can potentially be a better alternative for signful systems, if challenges in representation and optimization can be overcome. 
In one dimension, a fermion system can be mapped to a system of hard-core bosons via the Jordan-Wigner (JW) transformation \cite{JW}, which can then be variationally solved with matrix product states (MPS) as the wavefunction ansatz optimized via the density matrix renormalization group algorithm \cite{dmrg}.

In two dimensions, the JW transformation introduces non-locality to the system; an alternative wavefunction ansatz that preserves locality in fermion systems is more suitable for representing ground states of local Hamiltonians. 
Fermion projected entangled pair states (fPEPS) \cite{fPEPS_VF, fPEPS_QC, fPEPS_VF2, fPEPS_graded} are invented for this: given an fPEPS state, if a local operator supported in region $R$ acts on the state, only the tensors within $R$ change. 
The initial formulation of fPEPS is based on the virtual fermions living on the virtual bond of the PEPS. 
Later on, people have also adopted equivalent formulations based on swap gates \cite{fPEPS_SG, fMERA_SG} or Grassmann numbers \cite{fPEPS_Grassmann,liu2025hubbard} to study fPEPS.

Optimizing a PEPS for the ground states of strongly correlated systems \cite{verstraete2004renormalization, algorithm_finite_peps} is a long-standing challenge in the field. 
Recently, variational Monte Carlo (VMC) optimization has emerged as the method that gave the lowest computed ground state energies for several challenging systems \cite{liu2018gapless,liu2022gapless,liu2022emergence,liu2024emergent,liu2024j1j2j3,liu2024quantum}. 
In particular, it has given the state-of-the-art result for fermion Hubbard model at $\frac{1}{8}$ hole doping \cite{liu2025hubbard}, a particularly challenging and important regime. 
In the appendix of Ref. \cite{liu2025hubbard}, the VMC algorithm is explained for fPEPS in the Grassmann formulation.
However, the Grassmann formulation requires subtle reimplementations of tensor network operations, and it is not entirely clear how to adapt it to fPEPS implementations written in the swap gates formulation, which is one of the simplest formulations used in the field. Here we would like to derive and demonstrate the VMC algorithm in the swap gates formulation, and spell out the implementation of fermion sign in each step of the algorithm.

Crucially, as we explain the sampling algorithm, we also prove the detailed balance of the sequential sampling of PEPS \cite{liu2021}, which has remained a confusion in the field and left unproved until now.  

This paper is organized as follows. 
In section \ref{sec:swap_gates}, we explain fPEPS in the swap gates formulation. 
In section \ref{sec:vmc}, we explain the VMC algorithm of fPEPS and prove the detailed balance of the sequential sampling. 
In section \ref{sec:examples}, we demonstrate the result of VMC of fPEPS for certain representative fermion systems. 
In section \ref{sec:conclusion}, we discuss and conclude. 
\section{Fermion PEPS in the swap gates formulation}
\label{sec:swap_gates}
\subsection{spinless fermions}
The tensor product structure of the many-body fermion Hilbert space can be established by turning it into a (hard-core) boson Hilbert space via prescribing the fermion modes an order, equivalent to using a JW string.  
\begin{figure}[t]
    \centering
    \includegraphics[width=0.8\linewidth]{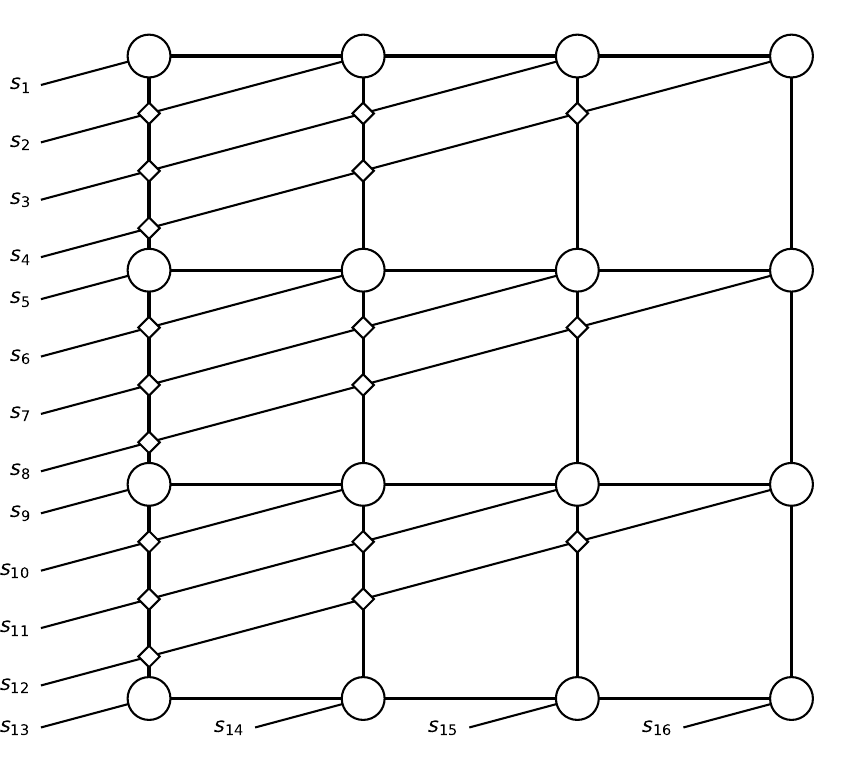} 
    \caption{A $4\times 4$ fermion PEPS with swap gates.}
    \label{fig:peps_sg}
\end{figure}
On a square lattice, in the swap gates formulation, the ordering of the fermion modes is shown in Fig. \ref{fig:peps_sg}. 
One first extends the physical legs of all tensors, towards bottom left, and terminate outside the network, and then orders the fermion modes starting from the top left tensor, counter-clock wise, until the bottom right tensor.  
This defines the computational basis of the many-fermion Hilbert space:
\begin{equation}
  \ket{\vec s} \equiv c_1^{\dag s_1} c_2^{\dag s_2} \cdots c_{N-1}^{\dag s_{N-1}} c_{N}^{\dag s_N} \ket{\vac}
  \label{eq:basis}
\end{equation}
where $s_i = 0$ or 1 specifying whether the $i$-th fermion mode is occupied or not.  
$N$ is the total number of complex fermion modes, and 
$\ket{\vac}$ is their common vacuum. 

An fPEPS is made of $\Z_2$-symmetric tensors, where each virtual leg, indexed with $k = 1, 2, \cdots, D$, is labelled with a parity charge $P_k = 0$ or 1. 
At every site $i$, one requires the $5-$leg fPEPS tensor to have even parity:  
\begin{equation}
  A[i]_{lrdu}^{s} = 0 
  \text{ if } P_l + P_r + P_d + P_u  \not= P_s \text{ mod } 2
  \label{eq:Z2}
\end{equation}
where $l, r, d, u$ are the indices for the left, right, down, and up virtual legs of $A$ at site $i$.
As a physical leg, $s$ is extended towards the bottom left, and every time it crosses another leg with index $k$, a swap gate is inserted:  
\begin{equation}
  \raisebox{-0.5\height}{\includegraphics[height=4.5em]{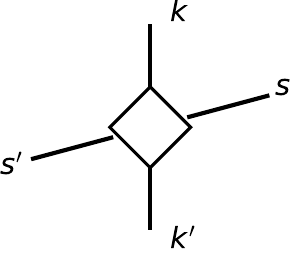}} = (-1)^{P_s P_k} \delta_{k k'}\delta_{s s'}
  \label{eq:sg}
\end{equation}
A fPEPS state $\ket{\Psi}$ is defined via its amplitude in the basis $\ket{\vec s}$:
$\braket{\vec s|\Psi}$ equals the contraction of the 2D swap-gates-decoreted tensor network, with physical indices projected to $\vec s$.  

The fPEPS defined above has two features: 
\begin{enumerate}
    \item Well-definedness: as one extends a physical leg, there are different ways by which virtual bonds are crossed by the physical leg.  
      The different ways are equivalent due to the gauge constraint Eq. \ref{eq:Z2}.   
    \item Locality: when acted on by a fermion operator in a local region, even if separated by a long JW string, e.g. $c_5^\dag c_9$ in Fig. \ref{fig:peps_sg}, only the tensors in that local region change. 
\end{enumerate}
See the appendix of Ref. \cite{fisoTNS} for a detailed explanation of the above. 
In our algorithm, we adopt the configuration drawn in Fig. \ref{fig:peps_sg}, where the physical legs are bent flat enough so that they only cross vertical virtual legs, and do not cross each other. 

Computing the amplitude $\braket{\vec s|\Psi}$ is the core of the VMC algorithm.
We explain how to do it in the next section.

\subsection{spinful fermions}
For spinful fermions, the local physical dimension is 4: $s_i  = (s_{i\up}, s_{i\down})$, and the fermion modes in $\ket{\vec s}$ are ordered as $s_{1\up}, s_{1\down}, s_{2\up}, s_{2\down}, \cdots$.
One can encode $\ket{00}, \ket{\uparrow0}, \ket{\uparrow\downarrow}, \ket{0\downarrow}$ as $s=0,1,2,3$ so that its parity $P_s = s \text{ mod } 2$.

\subsection{number conservation}
In the framework of sampling, a definite particle number $M$ 
can be imposed by only sampling over product states $\vec s$ with particle number equal to $M$.
This approach is equivalent to computing the property of $P_M\ket{\Psi}$ instead of $\ket{\Psi}$, where $P_M$ is the projector to the total number sector $M$.
Implementing the projector by sampling instead of tensors increases the representability of the variational ansatz over bare fPEPS.
We adopt number conservation this way in Sec. \ref{sec:examples}. 

\section{Variational Monte Carlo of fermion PEPS}
\label{sec:vmc}
In VMC, the expectation value of an operator $K$ is calculated as \begin{equation}
  \frac{\braket{\Psi|K|\Psi}}{\braket{\Psi|\Psi}} = \sum_{\vec s} \frac{\abs{\braket{\vec s|\Psi}}^2}{\braket{\Psi|\Psi}} \frac{\braket{\vec s|K|\Psi}}{\braket{\vec s|\Psi}} \equiv \sum_{\vec s} p(\vec s) K_\loc(\vec s)
\end{equation}
which is obtained via sampling $\vec s$ from the probability distribution $p(\vec s) = \frac{\abs{\braket{\vec s|\Psi}}^2}{\braket{\Psi|\Psi}}$, and averaging over the estimator $K_\loc(\vec s) = \frac{\braket{\vec s|K|\Psi}}{\braket{\vec s|\Psi}}$. 

To obtain $\ket{\Psi}$ as an approximate ground state of a Hamiltonian $H$, one performs optimization to minimize $\frac{\braket{\Psi|H|\Psi}}{\braket{\Psi|\Psi}}$ based on the gradients of $\braket{H}$ with respect to the complex conjugate \footnote{An fPEPS is a holomorphic ansatz.} (denoted by an overhead bar) of the parameters $\theta_\alpha$ in the wavefunction ansatz:
\begin{equation}
  g_\alpha \equiv \frac{\partial \braket{H}}{\partial \overline{\theta_\alpha}} = [\overline{O_{\alpha}(\vec s)} H_\loc(\vec s)] - [\overline{O_{\alpha}(\vec s)}][H_\loc(\vec s)]
\end{equation}
where $[\cdot]$ denote statistical average with respect to $p(\vec s)$. 
$O_{\alpha}(\vec s)$ is the log-derivative of $\braket{\vec s|\Psi}$ with respect to $\alpha$
\begin{equation}
  O_{\alpha}(\vec s) \equiv \frac{\partial \ln \braket{\vec s|\Psi}}{\partial \theta_\alpha} =  \frac{\partial\braket{\vec s|\Psi}/\partial \theta_\alpha}{\braket{\vec s|\Psi}}
\end{equation}

One thus needs four subroutines in a VMC calculation:  
\begin{enumerate}
  \item computing the amplitude $\braket{\vec s |\Psi}$
  \item computing the log-derivative $\frac{\partial\braket{\vec s| \Psi}/\partial \theta_\alpha}{\braket{\vec s| \Psi}}$  
  \item computing the local energy $\frac{\braket{\vec s|H|\Psi}}{\braket{\vec s|\Psi}}$ 
  \item a Monte Carlo update to sample $\vec s$ from $p(\vec s)$ 
\end{enumerate}
We now detail the above in the context of fPEPS. 
We consider an fPEPS of size $L_x \times L_y$ and bond dimension $D$. 

\subsection{fPEPS tensors after sampling}

For any given sample $\mathbf{s}$, the index of each physical leg at site $i$ is fixed to $s_{i}$. The 5-leg tensor $A[i]_{lrdu}^{s}$ then becomes a 4-leg tensor $\tilde{A}[i]_{lrdu} = A[i]_{lrdu}^{s_{i}}$ with a total $\mathbb{Z}_2$ charge equals $P_{s_{i}}$

\begin{align}
\tilde{A}[i]_{lrdu} = 0, \text{if } P_l + P_r + P_d + P_u \not= P_{s_{i}} \text{ mod } 2
\end{align}

Pictorially, we represent the tensor after sampling by changing the physical legs to dashed lines.
A dashed line is no longer an index of the tensor but a fixed number to book-keep its $\mathbb{Z}_2$ charge.
We draw it when it is needed for swap gates.

Correspondingly, for a given $s$, the swap gate reduces to a gate in the Hilbert space of the virtual leg
\begin{equation}
  \raisebox{-0.5\height}{\includegraphics[height=4.5em]{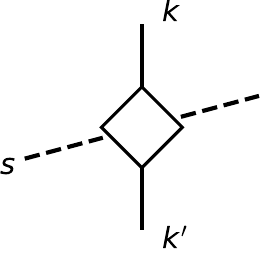}} = (-1)^{P_s P_k} \delta_{k k'}
  \label{eq:sg}
\end{equation}

\subsection{Computation of the amplitude and the boundary MPS environment}
The amplitude $\braket{\vec s|\Psi}$ can be evaluated using boundary MPS (bMPS) contraction.
This needs the bMPS left environments (\LE) or right environments (\RE). 
Both $\LE$ and $\RE$ are a list of bMPS with maximal bond dimension $D' > D$. 
$\texttt{\LE[0]}$ is taken as the trivial MPS, and $\texttt{\LE[x]}$ is then iteratively computed for $x=1$ to $L_x$ as
\begin{equation}
  \raisebox{-0.5\height}{\includegraphics[height=12em]{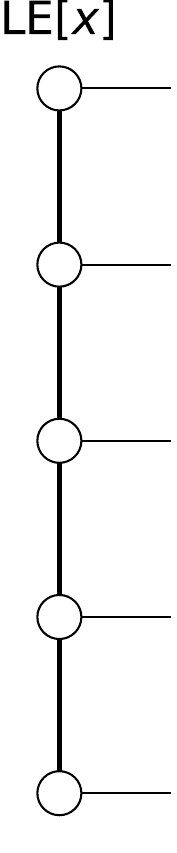}} \hspace{5mm} \approx \hspace{5mm} \raisebox{-0.5\height}{\includegraphics[height=12em]{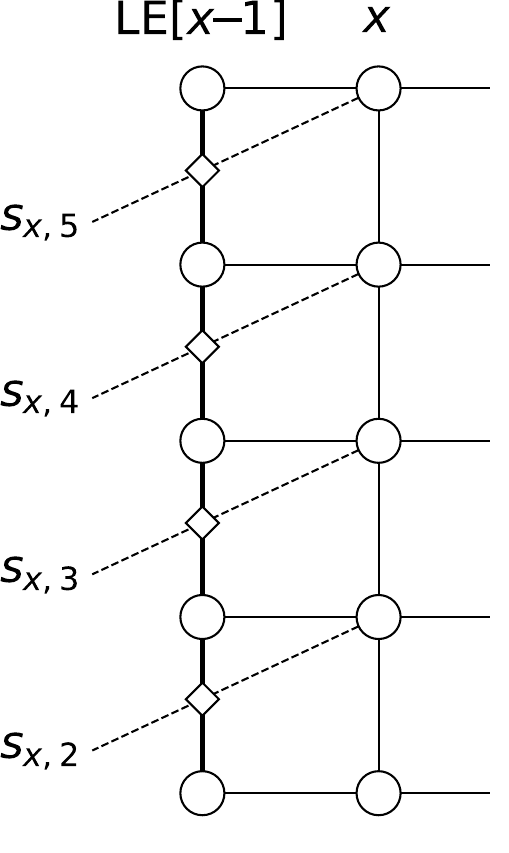}}
  \label{eq:LE}
\end{equation}
where $s_{x,y}$ is the value of the sample at the coordinate $(x,y)$.

Similarly, $\texttt{\RE[Lx+1]}$ is the trivial MPS, and $\texttt{\RE[x]}$ is iteratively computed for $x = L_x$ to $1$ as 
\begin{equation}
  \raisebox{-0.5\height}{\includegraphics[height=12em]{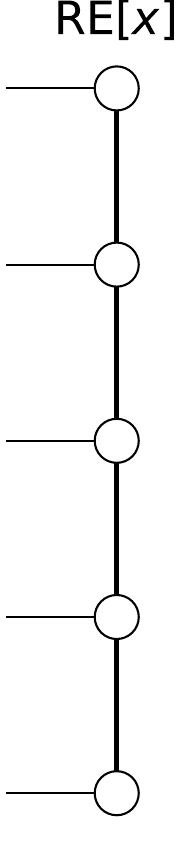}} \hspace{5mm} \approx \hspace{5mm} \raisebox{-0.5\height}{\includegraphics[height=12em]{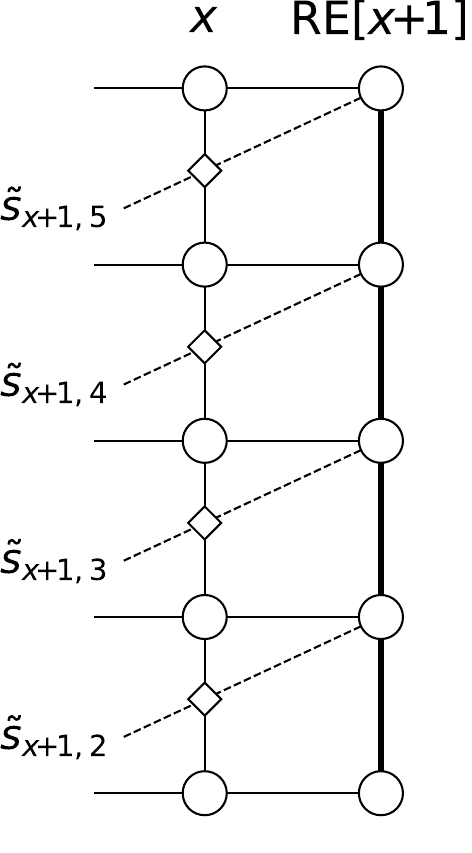}}
  \label{eq:RE}
\end{equation}
where $\tilde{s}_{x,y} \equiv (\sum_{x' \ge x} s_{x,y}) \mod 2$ is the sample assoicated with $\texttt{RE[x][y]}$.  
It can be updated as $\tilde{s}_{x,y} = (s_{x,y} + \tilde{s}_{x+1,y}) \mod 2$.

Eq.~\ref{eq:LE} and \ref{eq:RE} are both implemented via MPO-MPS variational compression with complexity $O(D'^3D^2 + D'^2 D^4)$ \cite{liu2021}.
The amplitude is then computed as
\begin{equation}
  \braket{\vec s|\Psi} = \hspace{5mm} \raisebox{-0.5\height}{\includegraphics[height=12em]{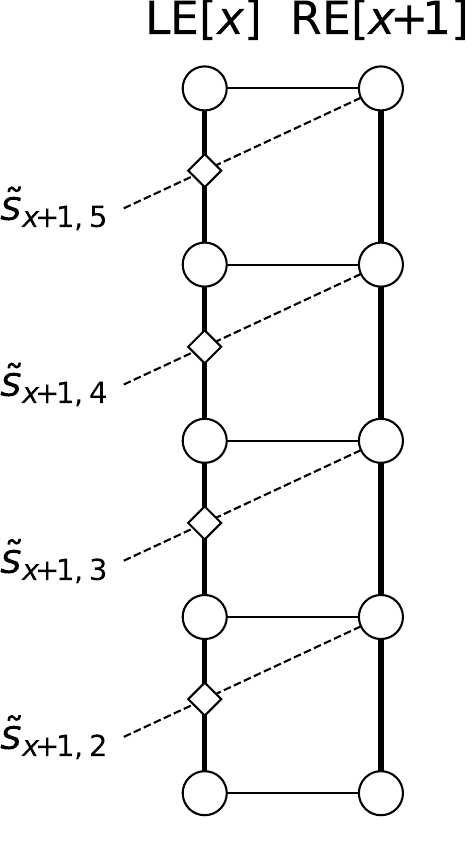}}
  \hspace{5mm} \text{for any $x$}
\end{equation}
\subsection{Computation of the log-derivative}
The variational parameters of an fPEPS are the $\Z_2$ symmetric tensors $A[i]_{lrdu}^{p_i}$ with $i$ labeling the sites.  
Given sample $\vec s$, the log-derivative of $\braket{\vec s|\Psi}$ with respect to $A[i]$ is 0 for $p_i \not= s_i$. 
For $p_i = s_i$, the log-derivative is computed as the contraction of the 2D network with a hole at site $i$:   
\begin{equation}
  \frac{\partial \ln \braket{\vec s|\Psi}}{\partial A_{lrdu}^{s_i}} = \raisebox{-0.5\height}{\includegraphics[height=10em]{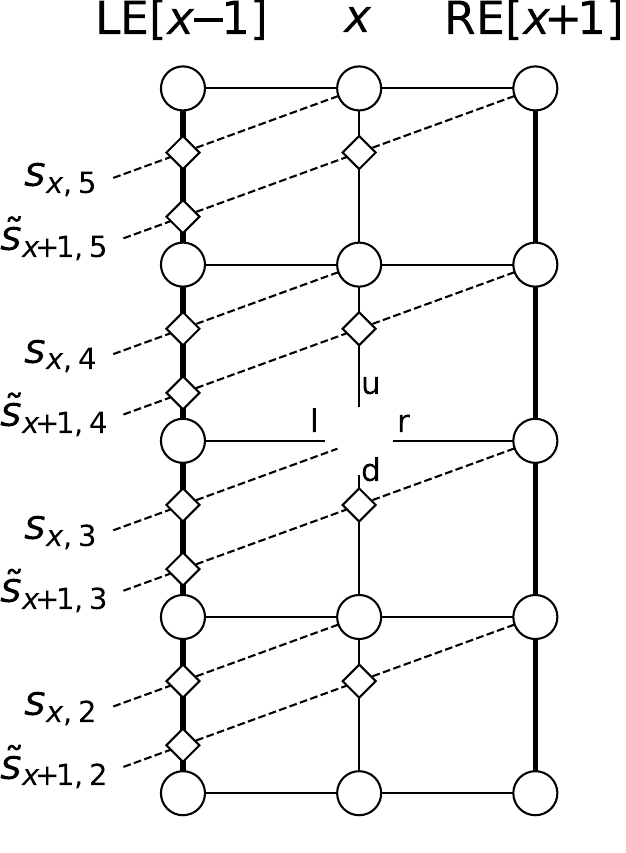}}  \divisionsymbol \hspace{2mm} \raisebox{-0.5\height}{\includegraphics[height=10em]{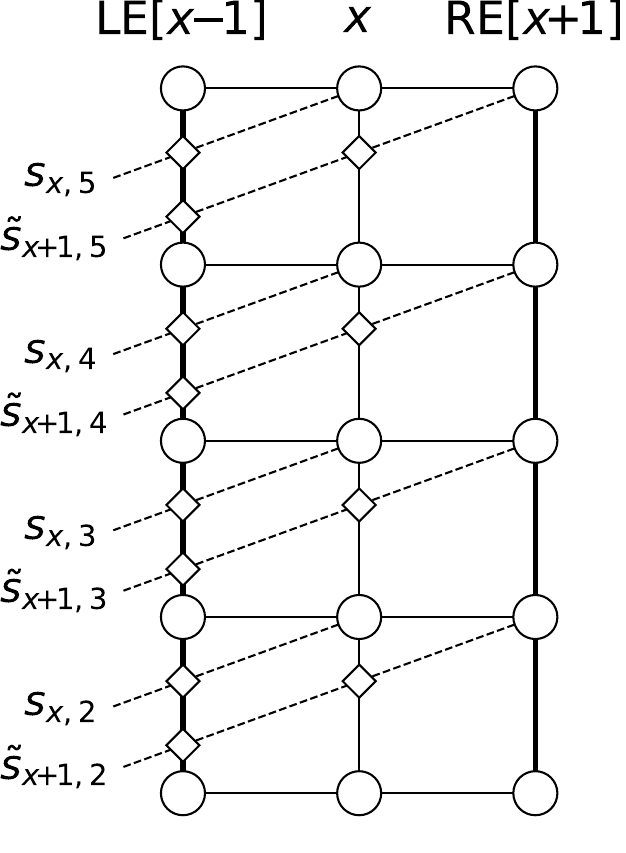}}
\end{equation}
This can be contracted exactly with complexity $O(D'^3D^2 + D'^2 D^4)$.

\subsection{Computation of the local energy}
Consider a short-range Hamiltonian with fermion hopping and an interaction which is diagonal in the basis $\mathbf{s}$. For diagonal terms $O$ in the Hamiltonian, the local energy is just its diagonal element $\frac{\braket{\vec s|O|\Psi}}{\braket{\vec s|\Psi}} = O_{\vec{s},\vec{s}}$. Thus, we only need to compute the local energy for the hoppings.
To compute $\frac{\braket{\vec s|c_i^\dag c_j|\Psi}}{\braket{\vec s|\Psi}}$, one acts the hopping to $\bra{\vec s}$, and account for the JW sign properly: $\bra{\vec s}c_i^\dag c_j = (-1)^{\#_{ij}}\bra{\vec s'}$, where $\#_{ij}$ is the number of fermions in $\vec s$ between (excluding) $i$ and $j$ in the JW string.
\begin{equation}
  \frac{\braket{\vec s|c_i^\dag c_j|\Psi}}{\braket{\vec s|\Psi}} = (-1)^{\#_{ij}}\frac{\braket{\vec s'|\Psi}}{\braket{\vec s|\Psi}} 
\end{equation}
where $\vec s'$ and $\vec s$ only differ locally and $\braket{\vec s|\Psi}$ and $\braket{\vec s'|\Psi}$ should be computed with the same $\LE$ and $\RE$. 
\subsection{Monte Carlo update}
The locality of an fPEPS allows extensive reuse of $\LE$ and $\RE$.
For any neighboring sites $i$ and $j$, one can perform a 2-site local update: if a particle occupies one and only one of the two sites, swap the occupation and propose the post-swap configuration $\vec s'$, and accept $\vec s'$ with probability $\min\{1, \frac{\abs{\braket{\vec{s'}|\Psi}}^2}{\abs{\braket{\vec s|\Psi}}^2}\}$; if not, do nothing and keep the current $\vec s$.

One MC sweep is defined as performing one 2-site local update for every nearest-neighbor vertical and horizontal bond on the lattice. 
The order of doing the local updates is sequential: one starts from the site at $(x,y) = (1,1)$, and then move up to $(1,2)$, until $(1,L_y)$; and then move to the next column at $(2,1)$, $\cdots$ until the top right corner of the lattice. 
This MC procedure is called {\it sequential sampling}, and was first proposed in Ref. \cite{liu2021}.
Its benefit is the ability to reuse $\LE$ and $\RE$ during the MC sweep. 
The complexity of one MC sweep is $O(L_x L_yD'^2D^4)$.

It is, however, not obvious that sequential sampling is reversible, i.e. satisfying detailed balance, due to an apparent sequence in the MC. 
This has remained a confusion thus far.
In the next section, we prove that sequential sampling satisfies detailed balance, and thus is guaranteed to converge to a unique stationary distribution $p(\vec s)$ if the MC update is also irreducible and aperiodic. 

\subsection{Detailed balance of sequential sampling}
Consider a system of $N$ spins or fermion modes. 
Let the sample space $\Omega$ be the product of $N$ local sample spaces $\Omega_i$. 
Let $\vec s  = (s_1, s_2, \cdots, s_N) \in \Omega$, where $s_i  \in \Omega_i$. 
Given the target probability distribution $\pi(\vec s)$, the sequential sampling procedure is defined as follow:  
\\
\\
\begin{algorithm}[H]
\caption{Sequential sampling}
\For{$i \gets 1$ \KwTo $N$}{
    Denote the spin configuration as $\vec s$\;
    Fix all spins excluding spin $i$: $\vec s_{\not= i}$\;
    Sample new spin $s'_i$ with transition probability $p_i(s_i \shortarrow s'_i | \vec s_{\not=i})$ that satisfies the conditional detailed balance, restricted to spin $i$:
    \vspace{-0.3cm}
    \begin{align*}
      p_i&(s_i \shortarrow s'_i | \vec s_{\not= i})\pi(\vec{s}_{1:i-1}, s_i, \vec{s}_{i+1:N}) \\
     &= p_i(s'_i \shortarrow s_i| \vec s_{\not= i})\pi(\vec{s}_{1:i-1}, s'_i, \vec{s}_{i+1:N})
  \end{align*}
}
\label{algo:sequential}
\end{algorithm}
\noindent where $\vec{s}_{a:b} \equiv (s_a, s_{a+1}, \dots, s_b)$.

\begin{theorem}
  Algorithm \ref{algo:sequential} (sequential sampling) satisfies detailed balance with $\pi(\vec s)$ as the stationary distribution.
\end{theorem}
\begin{proof}
  We take $N=3$ for notational simplicity. 
  It should be obvious that the proof holds for other $N$. 
  Consider $\vec s = (s_1, s_2, s_3)$ and $\vec s' = (s'_1, s'_2, s'_3)$ before and after one step of sequential sampling.  
  By the conditional detailed balance: 
  \begin{align*}
    p_1(s_1\shortarrow s'_1|s_2, s_3) \pi(\vec s) &= p_1(s'_1\shortarrow s_1| s_2, s_3) \pi(s'_1, \vec s_{2:3})
\\
p_2(s_2 \shortarrow s'_2 |s'_1, s_3) \pi(s'_1, \vec s_{2:3}) &= p_2(s'_2\shortarrow s_2 |s'_1, s_3) \pi(\vec s'_{1:2}, s_3)
\\
p_3(s_3 \shortarrow s'_3 |s'_1, s'_2) \pi(\vec s'_{1:2}, s_3) &= p_3(s'_3\shortarrow s_3 |s'_1, s'_2) \pi(\vec s')
  \end{align*}
Multiplying the three equations gives the detailed balance of sequential sampling: 
\begin{equation}
    p(\vec s \shortarrow \vec s') \pi(\vec s) = p(\vec s' \shortarrow \vec s) \pi(\vec s')
\end{equation}
where $p(\vec s \shortarrow \vec s')$ is the transition probability from $\vec s$ to $\vec s'$ in one step of sequential sampling.
\end{proof}
This proof easily generalizes to the case where the conditional sampling is used for multiple spins in a local region instead of a single site. 

\subsection{Stochastic reconfiguration}
To obtain the ground state, in the context of VMC, one often uses the method of stochastic reconfiguration (SR), which we also adopt in Sec. \ref{sec:examples}.  
SR is extensively discussed and documented in the literature \cite{SR1998}. 
Its application to fPEPS is no different. 
We document its detail here for completeness. 

Let $\vec \theta$ be a complex vector containing all of variational parameters in an fPEPS. 
Instead of doing steepest descent: $\theta'_\alpha = \theta_\alpha - \delta\cdot\frac{\partial \braket{H}}{\partial \overline{\theta_\alpha}}$, where $\delta$ is the descent step size (also called the {\it learning rate}), one descends along the direction $\vec X$, i.e. $\vec{\theta}' = \vec\theta - \delta \cdot \vec X$, where $\vec X$ solves the matrix equation  
\begin{equation}  
  \sum_\beta S_{\alpha\beta} X_{\beta} = g_\alpha = \frac{\partial \braket{H}}{\partial \overline{\theta_\alpha}}
\end{equation}
where $S_{\alpha\beta}$ is the quantum metric \cite{QGT}. 
It is the unique metric on the variational manifold under which the geodesic distance between two points is proportional to the 2-norm distance between the two corresponding quantum states in the Hilbert space.
It can be obtained via sampling:
\begin{equation}
  S_{\alpha\beta} = [\overline{O}_\alpha(\vec s) O_{\beta}(\vec s)] - [\overline{O}_\alpha(\vec s)][O_{\beta}(\vec s)]
\end{equation}
The quantum metric matrix $S$ is almost always ill-conditioned.  
In practice, $\vec X$ is obtained via solving $(S + \epsilon) \vec X = \vec g$, where $\epsilon$ is a small regulator to $S$. 
\section{Examples}
\label{sec:examples}
\subsection{Free Fermion Benchmark}
Here we consider several paradigmatic free fermion examples to benchmark the VMC algorithm of fPEPS: sublattice band insulator (Insulator), Dirac semimetal (Dirac), $\phi = 2\pi/3$ Hofstadter model with Chern number 1 (Chern), and a system with a Fermi surface (Fermi):
\begin{equation}
    H = -\sum_{\langle ij \rangle} t_{ij} c_i^{\dagger} c_j - \sum_{i} \mu_i c_i^{\dagger} c_i\\
\end{equation}
In all systems, $|t_{ij}| = 1$. 
For Chern, the phases of $t_{ij}$ set the flux. 
For Dirac, the vertical hoppings are 1 ($-1$) along even (odd) columns, while the horizontal hoppings are all $1$.
For Insulator and Fermi, all $t_{ij} = 1$. 
We set $\mu_i = \pm 1$ for the two sublattices for Insulator; $\mu_i = 0$ for other models.
For Insulator, Dirac, and Fermi, we fix the particle number at half filling.  
For Chern, we fix the particle number at 2/3 filling.

We list the hyper-parameters of the VMC algorithm in Table \ref{tab:parameters}.
Excessive fine tuning is not needed for these parameters.
\begin{table}[h]
\centering
\begin{tabular}{l@{\hskip 1cm}l}
\hline
\textbf{Parameter} & \textbf{Value} \\
\hline
fPEPS bond dimension $D$            & 4             \\
bMPS bond dimension $D'$      & 12 (20 for Fermi) \\
SR regulator $\epsilon$       & 0.001         \\
step size $\delta$            & 0.1           \\
number of samples                    & 4096          \\
tensor initialization                    & random in [0,1]        \\
\hline
\end{tabular}
\caption{Hyperparameters used in the VMC.}
\label{tab:parameters}
\end{table}

As shown in Fig. \ref{fig:10x10}, fPEPS VMC works very well for $10\times 10$ Insulator and Dirac, and $12 \times 12$ Chern, reaching energy density error on the order of $10^{-5}$.  
For Fermi, the calculation is more challenging as expected for its large entanglement and large density of states near the ground state. 
\begin{figure}[htbp]
  \centering
  \begin{subfigure}[t]{0.23\textwidth}
    \includegraphics[width=\linewidth]{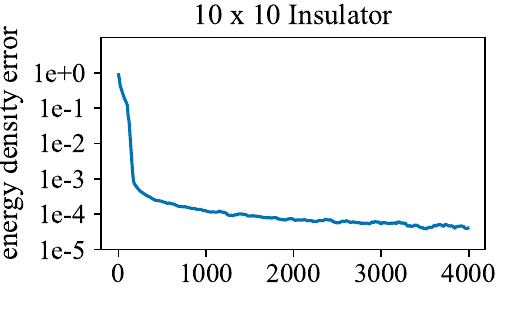}
  \end{subfigure}%
  \begin{subfigure}[t]{0.23\textwidth}
    \includegraphics[width=\linewidth]{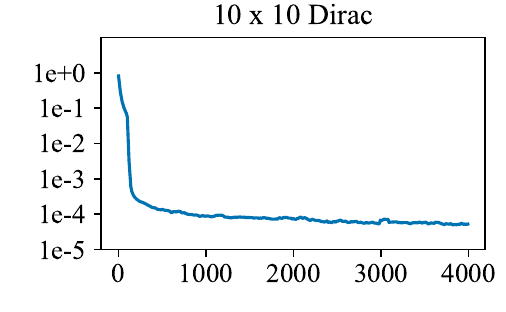}
  \end{subfigure}
  \begin{subfigure}[t]{0.23\textwidth}
    \includegraphics[width=\linewidth]{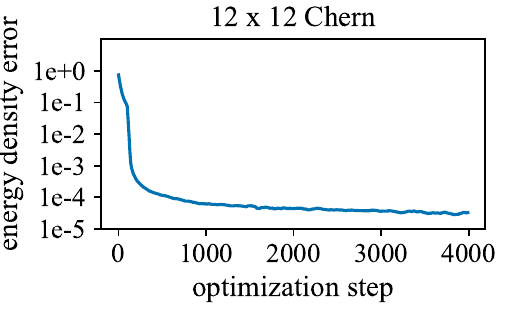}
  \end{subfigure}%
  \begin{subfigure}[t]{0.23\textwidth}
    \includegraphics[width=\linewidth]{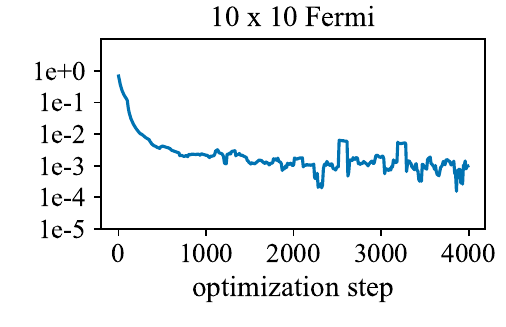}
  \end{subfigure}
  \caption{Energy density error $\frac{E-E_\text{exact}}{L^2}$ as the VMC optimization steps increase.
  The data are presented as the moving average over every 100 steps.
  At each step, the energy is obtained via averaging over 4096 samples.}
  \label{fig:10x10}
\end{figure}

An interesting comparison is with boson PEPS representing the ground state of the fermion systems.
Given a fixed computational basis, one may choose to represent the wavefunction amplitude of the fermion system as a boson PEPS, without the swap gates.
To find the ground state, one needs to perform a Jordan-Wigner transformation to get a boson Hamiltonian and run the bosonic algorithms.
In the VMC algorithm, the long-rangeness of the JW string does not pose any algorithmic difficulty in computing the energy gradient, since it is diagonal in the basis of sampling. This differs from traditional PEPS methods, such as simple or full update \cite{simple_update, full_update_1, full_update_2}.
For the VMC algorithm, one expects the difference between boson PEPS and fPEPS to come only from the different representing power of the ansatz.
Since the fPEPS ansatz maintains locality under local gate actions, we expect it to be a better ansatz for ground states of local fermion Hamiltonians, which is indeed what we observe in Fig.~\ref{fig:4x4}.
In fact, the performance of boson PEPS is so bad that we could not even converge the calculation at $10 \times 10$ on first tries, so we use $4\times4$ ($6 \times 6$ for Chern) as illustrating examples. 
We expect the difference to be even more pronounced for larger system size.

\begin{figure}[htbp]
  \centering

  \begin{subfigure}[t]{0.23\textwidth}
    \includegraphics[width=\linewidth]{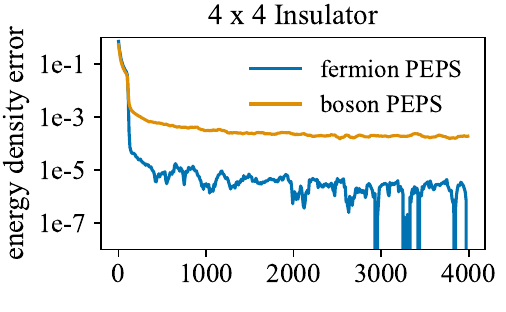}
  \end{subfigure}%
  \begin{subfigure}[t]{0.23\textwidth}
    \includegraphics[width=\linewidth]{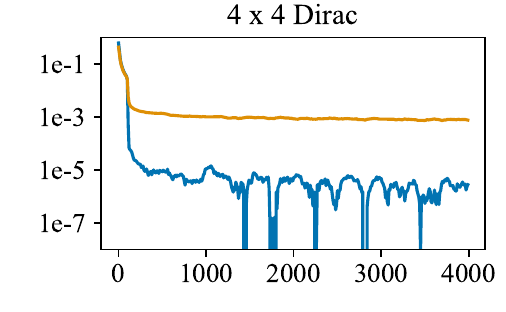}
  \end{subfigure}

  \vspace{1em}

  \begin{subfigure}[t]{0.23\textwidth}
    \includegraphics[width=\linewidth]{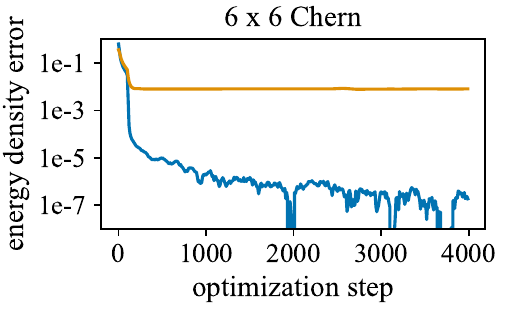}
  \end{subfigure}%
  \begin{subfigure}[t]{0.23\textwidth}
    \includegraphics[width=\linewidth]{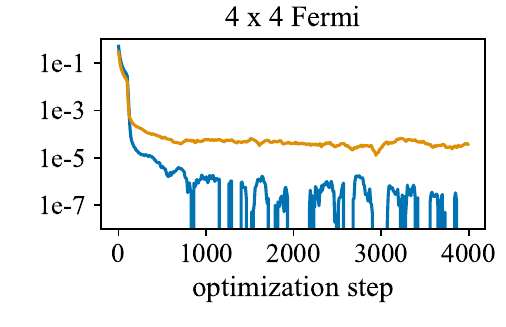}
  \end{subfigure}
  \caption{Comparison between fermion PEPS versus boson PEPS in terms of their performance in VMC of fermion systems. 
  Energy density error $\frac{E-E_\text{exact}}{L^2}$ is shown as the VMC optimization steps increase.
  The data are presented as the moving average over every 100 steps.
  At each step, the energy is obtained via averaging over 4096 samples.}
  \label{fig:4x4}
\end{figure}

\subsection{Hubbard Model Benchmark}
We also include a benchmark on the $4\times 4$ Hubbard model with $U=8$ at half-filling.
Its Hamiltonian is
\begin{equation}
  H = -\sum_{\braket{ij}, \sigma} c_{i,\sigma}^\dag c_{j,\sigma} + U \sum_i n_{i,\up} n_{i,\down}
\end{equation}
where $c_{i,\sigma}$ is the fermion operator on site $i$ with spin $\sigma = \uparrow$ or $\downarrow$.
We work at half-filling: $\sum_i n_{i,\up} = \sum_j n_{j,\down} = 8$.

The VMC results is shown in Fig. \ref{fig:hubbard}.  
As seen, the energy error is significantly larger than the free Fermion models, indicating large entanglement.
As a comparison, a recent fPEPS study using exact gradient optimization and double-layer contraction reports an energy density error of $4\times10^{-3}$ with $SU(2)$ symmetry at $D=6$, and $2\times 10^{-2}$ with $U(1)$ symmetry at $D=7$ \cite{scheb2023finite}. 

Another application of the current algorithm to the $8\times8$ Hubbard with next-nearest neighbor hopping at $1/8$ hope hopping is reported in \cite{gu2025solving} with bond dimension $D=10$, showing comparable accuracy to the state-of-the-arts neural quantum state ansatz \cite{gu2025solving}.     
\begin{figure}[htb]
  \centering
  \includegraphics[width=\linewidth]{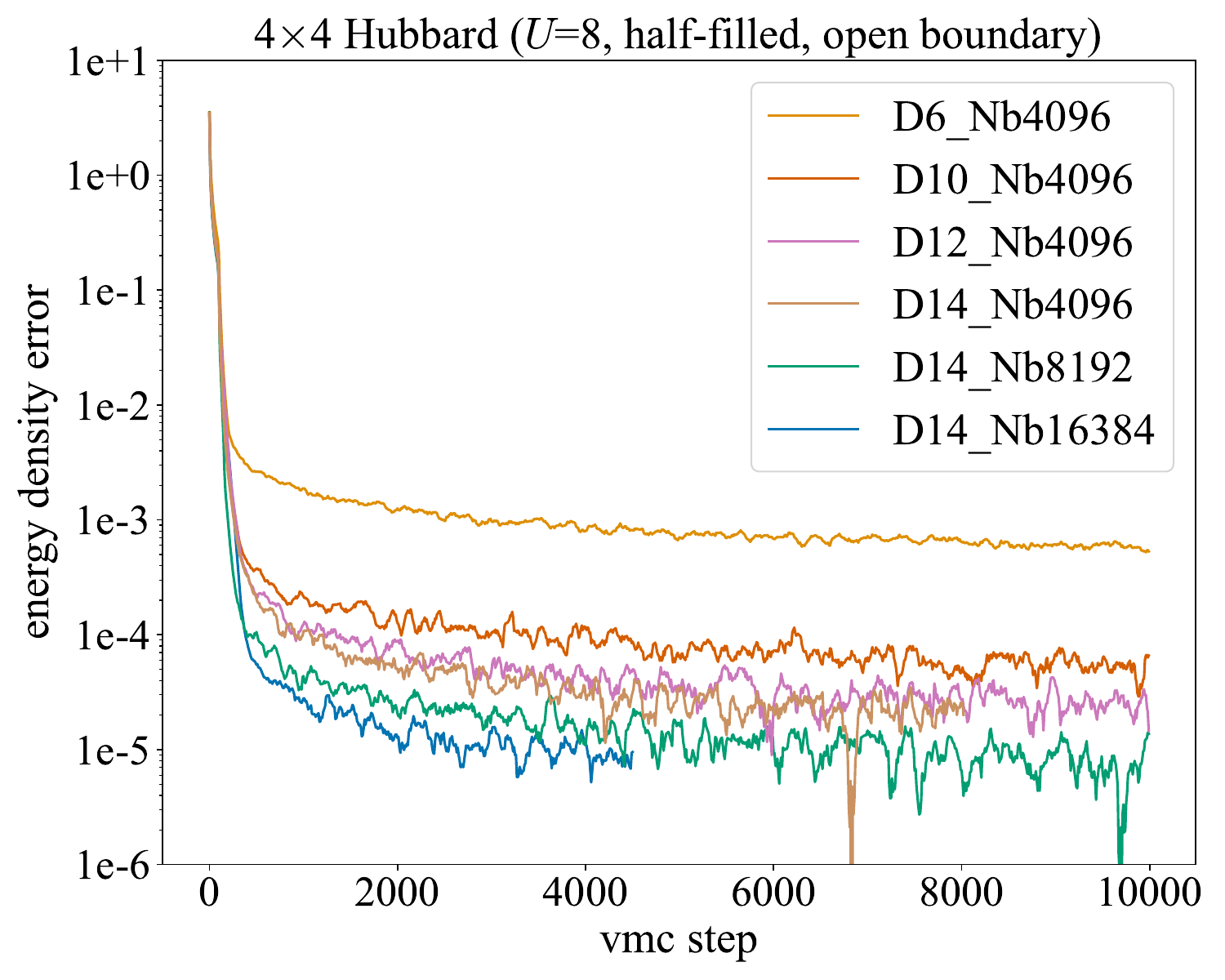}
  \caption{
  Energy density error $\frac{E-E_\text{exact}}{L^2}$ as the VMC optimization steps increase at various bond dimension $D$ and batch size $N_b$.
  $D'$ is taken to be $4D$; the other VMC hyperparameters are the same as in Table \ref{tab:parameters}.
  The data are presented as the moving average over every 100 steps.
  At each step, the energy is obtained via averaging over 4096 samples.
  $E_\text{exact}=-6.8084144664$ is found via exact diagonalization.
  }
  \label{fig:hubbard}
\end{figure}

\section{Conclusion} 
\label{sec:conclusion}
In this paper, we derived and explained the VMC algorithm for fermion PEPS in the swap gate formulation, showing its algorithmic simplicity and elegance. 
We also presented data demonstrating the effectiveness of the algorithm.
We proved the detailed balance for the sequential sampling, which has remained a confusion in the field.

The main future direction of this work is the application of this algorithm to challenging interacting fermion systems, such as the chiral spin liquid and superconductivity in Hubbard model. 
In addition to application, an important open question is to develop a framework where systems with a Fermi surface can be reliably simulated. 
This would require some non-trivial breakthrough from the PEPS structure, but should be possible in the flexible framework of VMC.

\begin{acknowledgements}
Part of the algorithm is tested using the TenPy code base \cite{hauschild2018efficient, johannes2024tensor}.
Y.W. is supported by a startup grant from the IOP-CAS.
Y.W. thanks Tao Xiang and Lei Wang for allowing him to use their computing resources.
The data that support the findings of this article are openly available \cite{data}; embargo periods may apply.
\end{acknowledgements}
\bibliography{ref.bib}
\end{document}